\documentclass[10, journal, compsoc]{IEEEtran}
\usepackage{subfigure}	
\usepackage{graphicx,tipa}
\usepackage{setspace}		
\usepackage{amsmath,amssymb,amsthm}
\usepackage{placeins}

\usepackage[lined,algonl,algoruled]{algorithm2e} 
\usepackage{algorithmic}

\newtheorem{theo}{Theorem}
\newtheorem{lemma}{Lemma}

\newtheorem{observation}{Observation}

\newtheorem{definition}{Definition}
\newcommand{\remove}[1]{}

\newcommand{\arc}[1]{{%
  \setbox9=\hbox{#1}%
  \ooalign{\resizebox{\wd9}{\height}{\texttoptiebar{\phantom{A}}}\cr#1}}}

\title{Approximation Algorithms for Road  Coverage Using Wireless Sensor Networks for Moving Objects Monitoring}
\author{

\IEEEauthorblockN{Dinesh Dash\\}
\IEEEauthorblockA{National Institute of Technology Patna, India. }

}

\date{}

\begin{document}

\IEEEtitleabstractindextext{%
\begin{abstract}


Coverage problem in wireless sensor networks measures how well a region or parts of it is sensed by the deployed sensors. Definition of coverage metric depends on its applications for which sensors are deployed. In this paper, we introduce a new quality control metric/measure called {\em road coverage}. It will be used for measuring efficiency of a sensor network, which is deployed for tracking moving/mobile objects in a road network. A {\em road segment} is a sub-part of a road network.  A {\em road segment} is said to be {\em road covered} if an object travels through the entire road segment then it must be detected somewhere on the road segment by a sensor. First, we propose different definitions of road coverage metrics. Thereafter, algorithms are proposed to measure those proposed road coverage metrics. It is shown that the problem of deploying minimum number of sensors to {\em road cover} a set of road segments is NP-hard. Constant factor approximation algorithms are proposed for {\em road covering} axis-parallel road segments. Experimental performance analysis of our algorithms are evaluated through simulations.

\end{abstract}
}
\maketitle

\section{Introduction}
\label{sec:intro}

In wireless sensor network (WSN), {\em coverage problem} is an important issue. Different measures of coverage exist depending on application of the sensor network. For example {\em area coverage} \cite{thai2008} verifies every point of a region and checks whether these points are under the sensing range of at least one sensor.  A set of target points are monitored by at least $k$ sensors in {\em target $k$-coverage} problem \cite{liu06}. In $k$-barrier-coverage problem \cite{kumar05}, all crossing paths across a boundary of a region are covered/sensed by at least $k$ sensors. In point sweep coverage, a set of points are monitored after a certain fixed time interval \cite{gorain_2:2014}.

Sensor deployment strategy to attain certain level of coverage is another area of research in sensor networks. Researchers are proposing different sensor deployment strategy with minimum cost to ensure desired level of coverage. Sensor deployment strategy for area coverage is proposed by Kim et al. \cite{kim:2008}. Minimum cost based deployment scheme to ensure target coverage is proposed by Xu et al. in \cite{xu2007}.  Similarly, Yick et al. \cite{yick:2004} proposed strategies for the placement of minimum number of beacons and data loggers for a given sensor network. For tracking moving/mobile objects few path coverage metrics are defined in \cite{harada2009,ram2007} and their analytical expressions are evaluated for a given random deployment. \emph{Track coverage} problem is addressed by Baumgartner et al. \cite{baumgartner:2008}, their objective is to place a set of sensors in a rectangular region to detect tracks by at least a given number of sensors.


\begin{figure}[h]
\centering
\includegraphics [width=7cm]{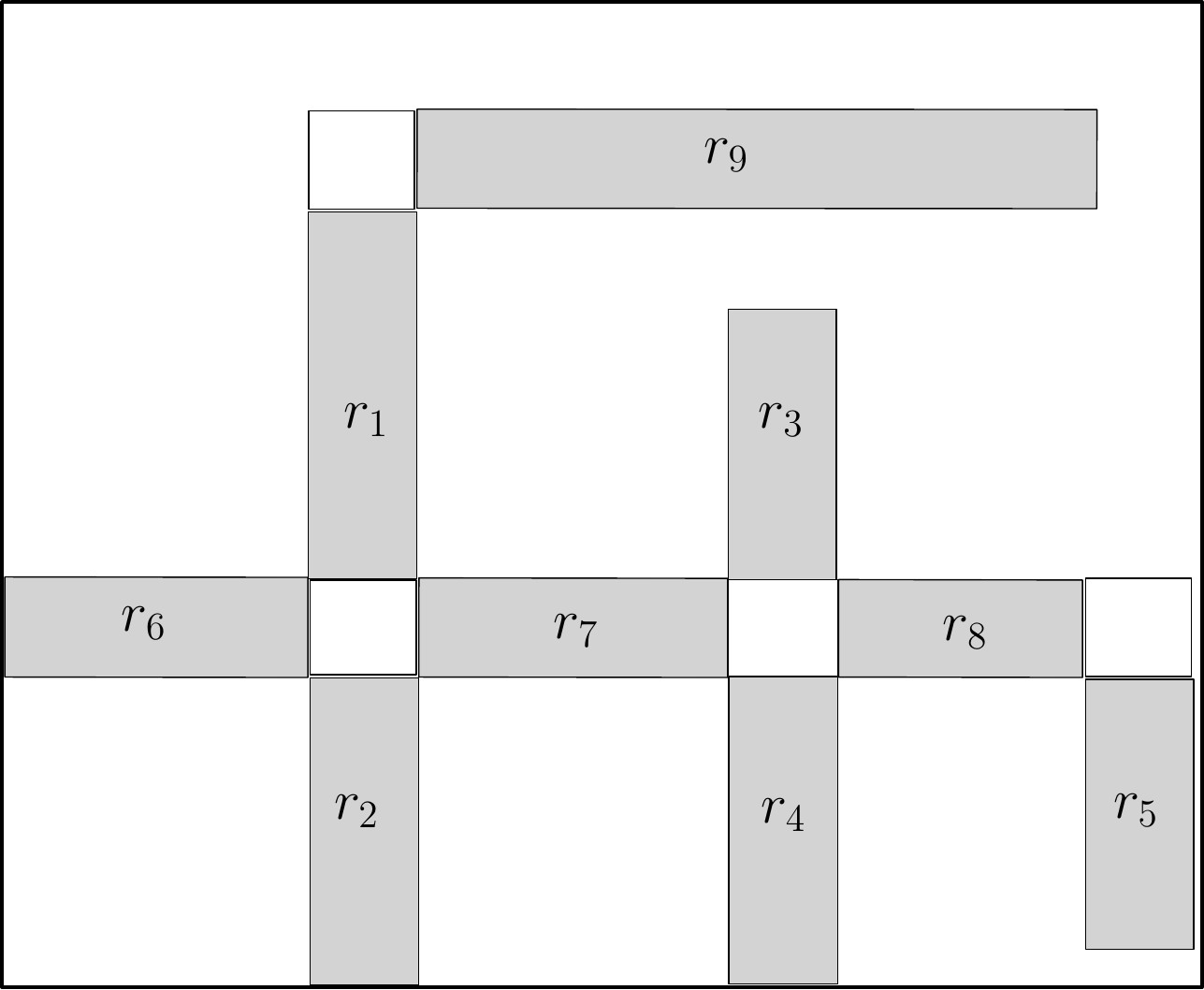}
\caption{A sample corridor or road network}
\label{road_netowrk}
\end{figure}

In certain applications, it is more suitable to cover some parts of a road network rather than the entire road network. Example of such applications are vehicles tracking, speed monitoring, congestion monitoring etc. in road networks. Although monitoring the entire road network is desirable but it makes the system expensive. Thus, the road network is divided into road segments as per requirements. And, in some cases it is acceptable to ensure that a vehicle entering one end and exiting the other end of a road segment should be monitored at least once somewhere in the road segment. An example of such road network and its road segments are shown in Figure \ref{road_netowrk}.  Let say, the road network is divided into a set of  shaded rectangular regions $(r_1,r_2,\ldots r_9)$. In general road segments are part of road networks, which is defined as per applications requirements. And our objective is to monitor vehicles/ objects traveling through the road segments. Gorain et al. in \cite{gorain:2014}, proposed scheme for patrolling a set of line segments using mobile sensors. They referred it as line sweep coverage and proposed approximation algorithm. In \cite{dash:2014}, different line coverage measuring schemes are proposed as smallest $k-covered$ line segment and longest $k-uncovered$ line segment. In this work, we generalize line segment to road segment, which is of rectangular shape and propose new coverage measures for it. In addition, we also propose different sensor deployment schemes to ensure quality of road coverage.

In this paper, we address this problem by defining a new measure of coverage called {\em road coverage} and its two variations {\em independent road coverage} and {\em collaborative road coverage} by sensing some part (length wise partial but width wise full) of a set of road segments, and propose algorithms for measuring {\em road coverage} and  sensors deployment schemes to achieve {\em road coverage}. We also show that the decision version of the sensor deployment problem is NP-hard, and present constant factor approximation algorithms for some special cases.


In summery main contributions of this paper are as follows:

\begin{itemize}
\item Proposed new coverage metrics called {\em road coverage} for tracking mobile objects moving on road networks. It ensures that if an object travels the full length of any road segment then it must be detected by the sensor network.

\item Proposed algorithms for measuring {\em road coverage} for a given a road networks.

\item Analyze complexity of sensor deployment problem to ensure {\em road coverage}.

\item Proposed sensor deployment algorithms for different special cases to ensure {\em road coverage}.
\end{itemize}


The rest of the paper is organized as follows. Section \ref{sec:relwork} briefly discusses related works on sensor coverage and deployment schemes. Section \ref{sec:prelim} presents some necessary backgrounds and the computational hardness of the problem. Section \ref{sec:measuring_road_coverage} describes our road coverage measurement algorithms. Sections \ref{sec:deployment_scheme} presents two constant factor approximation algorithms for deployment of sensors to ensure {\em road coverage}.  Experiment and result analysis are discussed in section \ref{sec:simulation}. Finally, section \ref{sec:conclude} concludes the paper and discusses some possible future extensions.

\section{Related works}
\label{sec:relwork}



In literature, different coverage measures are defined to compute various quality of coverage for a given sensor deployment. Huang and Tseng \cite{huang2005} propose algorithm to test whether a given bounded region is $k$- area-covered or not. They prove that if the perimeters of all the sensors within the bounded region are $k$-covered by their neighbors then the whole area is also $k$-covered.  Kumar et al. \cite{kumar05} propose a  coverage measure called barrier coverage and proposed algorithms to verify whether a given deployment ensures barrier coverage for a given boundary. They also proved that barrier coverage can not be determined locally. {\em Trap coverage} metric is proposed by Balister et al. \cite{balister2009}. It is the longest displacement an object can make in straight line within the target region without going inside the sensing range of any sensor or it is the diameter of the longest uncovered region within the target region. {\em Path coverage} is defined  in \cite{ram2007, harada2009} for tracking objects moving in straight line path. Dash et. al. in \cite{dash:2014} proposed deterministic algorithms for finding longest k-uncovered and smallest k-covered straight line path for mobile object within a bounded region. Garain et. al. \cite{Gorain:2014} propose {\em line sweep coverage} to ensure all points on a set of line segments are traverse by a set of mobile nodes within a fixed time interval. Point sweep coverage of a set of points is proposed in  \cite{gorain_2:2014}. They proposed both centralized and distribute algorithm for point sweep coverage. They extend the algorithm for point sweep coverage to area coverage.  The area is subdivided into squares of same size, which is dependent on the sensing range of the sensor such that if mobile sensors reaches the centre it can sense the complete square region. Now centre of the squares are considered as target points and apply the point sweep coverage algorithm on this set of points to ensure area sweep coverage. Baste et al. \cite{Baste:2017} introduce edge monitoring problem. A vertex $v \in V$ monitors an edge $ \{a, b \} \in E$  if $\{v, a \} \in E$ and  $\{v, b \} \in E$. Edge Monitoring problem finds a set $S$ of vertices of a graph of size at most $k$ such that each edge of the graph is monitored by at least one element of $S$. 


Finding a suitable deployment strategy to achieve desired level of coverage is another challenge in wireless sensor network. Kim et al. \cite{kim:2008} propose sensor deployment strategy to ensure 3-coverage of the deployed region as well as  sensors maintain a minimum separation among themselves.  Galota et al. \cite{galota:2001} propose a wireless base stations deployment scheme to cover a given set of target points such that the positions of the base stations are restricted to a finite set of feasible positions. Deployment scheme for covering a set of grid points is proposed by Chakrabarty et al. in \cite{chakrabarty2002}.  Xu et al. \cite{xu2007} proposed a minimum cost based deployment scheme to ensure target coverage where the position of the sensors and the target points are predefined. They assume the communication range of the sensors are sufficiently large such that they can communicate directly to the base station.  Wu et al. \cite{wu2007} propose a sensor deployment strategy in obstacle free region to maximize area coverage  by the deployed sensors. Clouqueur et al. \cite{clouqueur2002} propose a deployment strategy to ensure minimum exposure path for moving targets with minimum deployment cost.  Kumar et al. \cite{kumar05} provide optimal deployment strategy to ensure k-barrier coverage.  Bai et al.  propose an optimal sensor deployment strategy for ensuring connected coverage of a given area \cite{bai:2006}.  Agnetis et al. \cite{Agnetis:2009} address the problem of deploying sensors for full surveillance of a line segment with minimum cost under a defined cost model. They proposed a polynomial time optimal deployment scheme for covering a line segment using homogeneous sensors.  But, covering line segments using non-homogeneous sensors is NP-hard. They propose a branch-and-bound algorithm and a heuristic algorithm for non-homogeneous sensors. Zhang et al. \cite{zhang:2011} radars placement problem. Radars  are deployed on the banks of river which is modeled as piece-wise line segments. Radars are deployed to cover a given set of points on the river such that the total power consumption by the radars is minimum. Dash et. al. in \cite{dash:2013} propose deterministic sensor deployment schemes to ensure line coverage. Garain et. al. \cite{Gorain:2014} proposed deterministic algorithm for patrolling a set of line segments to ensure {\em line sweep coverage}.  In \cite{Njoya:2017} a stochastic optimization algorithm is proposed for sensor node placement to ensures target coverage with less sensors.


\section{Background and Problem Statements}
\label{sec:prelim}

In this section, necessary preliminary backgrounds and  problem statements are presented. We assume that the sensors are points in the plane and their sensing regions are circular disks. Let $circle(s_i,\rho)$ represent circular sensing region of sensor $s_i$ with sensing range $\rho$. Sensor $s_i$ can sense an event inside $circle(s_i,\rho)$.

\begin{definition}{\bf [Road Segment ($r_i=(l_i^t,l_i^b, w)$) :] }
A road segment  $r_i$ of width $w$ is a sub-part of a road network, which is defined by a rectangular region with two equal length parallel line segments ($l_i^t$, $l_i^b$) and their perpendicular separation $w$.
\end{definition}

\begin{figure}[h]
\centering
\includegraphics [width=7cm]{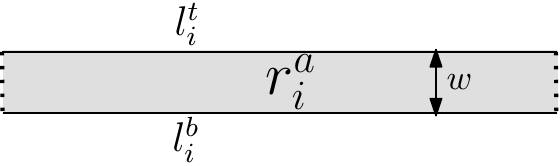}
\caption{Different parts of road segment $r_i$}
\label{road_segment}
\end{figure}

An example of road segment $r_i$ of width $w$ is shown in Figure \ref{road_segment}. $l_i^t$ and $l_i^b$ are referred as top and bottom {\em side boundary} of $r_i$. Apart from these two side boundaries, there are two more boundaries of length $w$, which are referred as {\em left end boundary} and {\em right end boundary}.  Let $r_i^a$ represent the rectangular region of the road segment $r_i$. Given a {\em road network}, which is partitioned  into set of road segments as per requirement and represent them by a set of road segments $R$. And given a  set of sensors $S$ and their sensing circles/disks. Based on the number of sensors independently or collectively sense/cover a road segment,  two {\em road coverage} metrics are proposed.

\begin{definition}{\bf [Independent Road Covered :] }
A road segment $r_i \in R$ is said to be {\bf independent road covered} by the sensors in $S$ if $r_i$'s full width but some part of its length is covered/sensed {\bf independently by a sensor} $s_j \in S$ so that if any object travels the full length of $r_i$ then it  must be detected by the sensor $s_j$.
\end{definition}

\begin{definition}{\bf [Collaborative Road Covered  :] }
A road segment $r_i \in R$ is said to be {\bf collaborative road covered } by the sensors in $S$  if its full width and some part of its length is sensed {\bf collectively by a  subset of sensors} $S_i^j \subseteq S$  so that if any object travels the full length of $r_i$ then the object must be detected by at least one sensor in $S_i^j$.
\end{definition}


 A road segment {\em independent road covered} implies it is also {\em collaborative road covered}. But, the reverse is not always true. Figure \ref{road_cover} shows an example of a sensor network in which a set of road segments $R= \{r_1, r_2,  r_3 \}$ are shown with shaded rectangles; a set of sensors $S=\{s_1, s_2, \ldots s_9\}$, whose  sensing regions are represented by unit circles/disks. In the figure road segment $r_1$ is {\em independent road covered} by sensor $s_3$ while the road segment $r_2$ is not {\em independent road covered} by any sensor but {\em collaborative road covered} by $\{s_7, s_5\}$. But road segment $r_3$ is neither independently nor collaboratively {\em road covered} by the deployed sensors.


\begin{figure}[h]
\centering
\includegraphics [width=7cm]{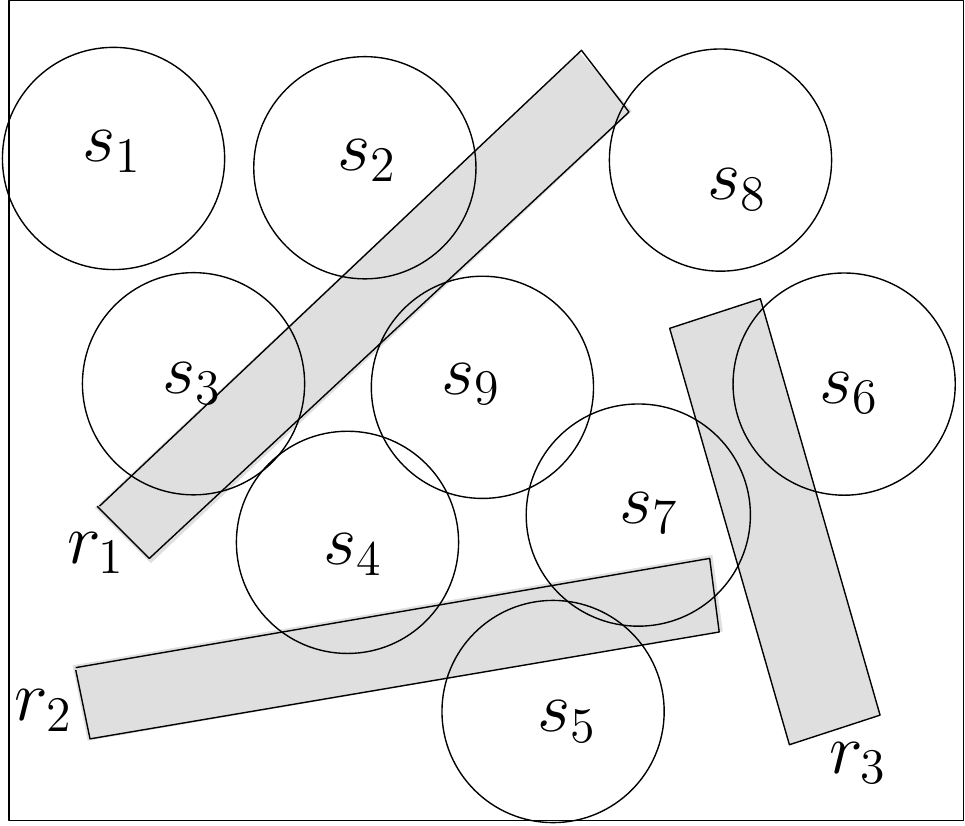}
\caption{An example of road coverage using sensors}
\label{road_cover}
\end{figure}


\begin{definition}{\bf [Independent Road coverage :] } A set of road segments $R =\{  r_1, r_2 \ldots r_n \}$  achieves {\em independent road coverage} by a given sensor deployment if and only if all the road segments in $R$ are independent road covered by the deployed sensors.
\end{definition}

\begin{definition}{\bf [Collaborative Road Coverage :] } A set of road segments $R =\{  r_1, r_2 \ldots r_n \}$  achieves {\em collaborative road coverage} by a given sensor deployment if and only if all the road segments in $R$ are collaborative road covered by the deployed sensors.
\end{definition}


{\bf Road Coverage Measure Problem :} Given a set  of $n$ road segments $R$ and  a set of $m$ sensors positions and their circular sensing circles. Verify {\em road coverage} (independent, collaborative) of the road segments in $R$. 

Once the algorithm to measure road coverage is known, subsequently our next objective is to deploy sensors to achieve road coverage. Formally, the problem can be stated as follows:

{\bf Road Coverage Deployment Problem (RCDPL) :} Given a set  of $n$  road segments $R$, use minimum number of sensors and find their positions such that all the road segments in $R$ are {\em road covered} (independent, collaborative ).

In the next section, algorithms for measuring road coverage are presented.  In subsequent section, sensor deployment algorithms to ensure road coverage are described. 

\section{Measuring Road Coverage}
\label{sec:measuring_road_coverage}


In this section, we present polynomial time algorithms to verify whether a given deployment of sensors ensures road coverage for a given road network. We assume that the sensing regions of the sensors are disks (may be of different sensing ranges).

\subsection{Independent Road Coverage Verification Algorithm}

In this sub section, we present an algorithm to check each road segment against the existing sensors and ensure whether all the road segments are {\em independent road covered} or not.


\begin{lemma} \label{}
A road segment $r_i$ is {\em independently road covered} by a sensor $s_j$ if and only if the circle corresponding to $s_j$'s sensing region intersects both bounding segments $l_i^t$ and $l_i^b$ of $r_i$.
\end{lemma}

\begin{proof}
If both bounding segments of $r_i$ are not intersected by a sensing disk of any senor $s_j$ then there exists a path for a mobile object that can traverse the full road segment without getting detected by $s_j$. Now, assume both bounding segments are intersected by $s_j$'s sensing circle. As both road segment area $r_i^a$ and the sensing circle $circle(s_j,\rho)$ are convex shapes therefore, the intersection between them is also a single convex region. Hence, if $s_j$'s sensing circle intersects both $l_i^t$ and $l_i^b$ of $r_i$ then the intersection region must contains some part of both $l_i^t$ and $l_i^b$ of $r_i^a$. And there always exists a line joining $l_i^t$ and $l_i^b$ through the intersection region, which is a part of full width of $r_i$. Hence, $s_j$ covers some part of $r_i$'s length and full width of $r_i$.
\end{proof}

Our algorithm for verifying {\em independent road coverage} is based on the above lemma. For a road segment $r_i$ consider its upper and lower bounding line segments ( $l_i^t$ and $l_i^b$ ) separately and determine the sensors whose sensing circles intersect both bounding segments. If there exists a sensor $s_j \in S$ whose sensing circle intersects both the segments then the road segment $r_i$ is said to be {\em independent road covered} by the sensor $s_j$. Same method is followed for all road segments in $R$ and if all road segments are {\em independent road covered} then the road network attains {\em independent road coverage} by the sensors in $S$. 

\begin{theo}
 Verifying {\em independent road coverage} for a road network with $n$ road segments and $m$ sensors can be done in $O(nm)$ time.
\end{theo}
 
\begin{proof}
For a particular road segment verifying independent road covered can be done in $O(m)$ time. Therefore, verifying $n$ road segment can be done in $O(nm)$ time.
\end{proof}


\subsection{Collaborative Road Coverage Verification Algorithm}

In this case,  one or more than one sensors together cover a particular road segment. If a road segment is not independently road covered by any deployed sensor then it may be collectively covered by more than one sensors.

In Figure \ref{coll_road_cover}, rectangular region denotes a road segment and circles denote sensors sensing regions. Three paths ($P_i^1$, $P_i^2$, $P_i^3$) are shown from top side boundary to bottom side boundary of the road segment. All of them are completely inside sensor's sensing regions but only $P_i^2$ is within the intersection of the road segment and sensors' sensing regions, whereas path $P_i^1$ and $P_i^3$ are not. Therefore, only path $P_i^2$ ensures collaborative road coverage of the road segment $r_i$, which is collectively sensed by sensors $\{s_4, s_5, s_6\}$.

\begin{figure}[h]
\centering
\includegraphics [width=8cm]{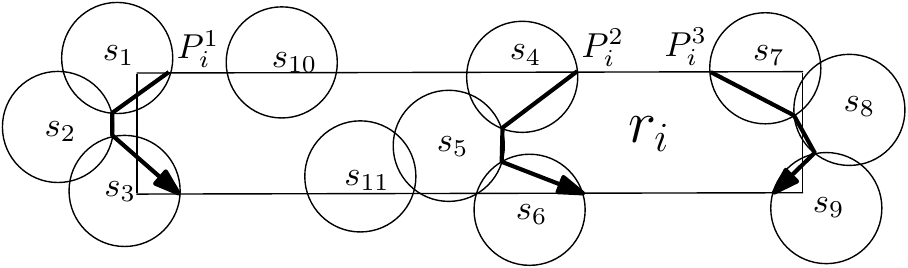}
\caption{Collaborative road coverage of road segment $r_i$}
\label{coll_road_cover}
\end{figure}

\begin{lemma} \label{lem:algo2:lem1} A set of sensors $S_i^j$ collaboratively road cover a road segment $r_i$ if and only if there exists a path $P_i^j$ from top side boundary segment $l_i^t$ to the bottom side boundary segment $l_i^b$ of the road segment $r_i$ which is completely inside the intersection regions of $r_i^a$ and the sensing circles of the sensors in $S_i^j$.
\end{lemma}

\begin{proof}
Assume there is no such path from top side boundary segment  $l_i^t$ to the bottom side boundary segment  $l_i^b$ of the road segment $r_i$, which is completely inside the intersection of the road segment and sensors sensing regions and the road segment attains {\em road covered}. Therefore, there exists a path inside $r_i^a$ through which an object can able to move the full length of the road segment $r_i$ without getting sensed/detected by any sensor. Hence, the road segment $r_i$  is not road covered by the sensors. It contradicts our assumptions. 
\end{proof}

Hence to verify {\em collaborative road covered} of a road segment,  a path is determined between top side boundary to bottom side boundary of the road segment, which is completely inside the intersection of the road segment and sensors' sensing regions. If such path exists then the sensors, which are covering the path, are able to detect objects moving through the road segment. The basic idea to measure collaborative road coverage for a road segment $r_i$ is as follows. 

For each road segment $r_i$ determine a set of sensors  $S_i^t$ whose sensing regions intersect the top side boundary $l_i^t$. Similarly, determine set of sensors  $S_i^b$ whose sensing regions intersect the bottom side boundary $l_i^b$. Construct an {\em intersection graph/ coverage graph} among sensors sensing regions and the road segment $r_i$. Let sensor $s_i$ be represented by a vertex $v_i$. There is an edge between two vertices $v_i$ and $v_j$ if $circle(s_i,\rho) \cap circle(s_j,\rho) \cap r_i^a \ne \phi$ where $r_i^a$ represents the area of the road segment $r_i$ and  $circle(s_i,\rho)$ represents   sensing region of $s_i$. For road segment $r_i$ consider two dummy vertices $v_i^t$ and $v_i^b$. Put edges between $v_i^t$ to all vertices corresponding to sensors in $S_i^t$ and from $v_i^b$ to all vertices corresponding to sensors in $S_i^b$. Once the {\em intersection graph} is determined for road segment $r_i$, determine a path between $v_i^t$ to $v_i^b$ in the intersection graph. Repeat the same process for all road segments in $R$.  Let  $PV_i^j$ denote a path between $v_i^t$ and $v_i^b$ on the intersection graph of $r_i$. Let $V_i^j$ denote  set of internal vertices on the path $PV_i^j$ except the start and end dummy vertices $v_i^t$ and $v_i^b$. Let $S_i^j$ denote set of sensors corresponding to the vertices in $V_i^j$.


\begin{lemma} \label{}
If there is a path $PV_i^j$ on the intersection graph of $r_i$ between the dummy nodes $v_i^t$ and $v_i^b$ then the road segment $r_i$ is collaboratively road covered by the sensors in $S_i^j$. 
\end{lemma}

\begin{proof}
In other words, if there is a path $PV_i^j$ between $v_i^t$ and $v_i^b$ on the intersection graph of $r_i$ then there exists a piecewise-linear path $P_i^j$  between $l_i^t$ and $l_i^b$.  In addition, the path $P_i^j$ is completely inside the  intersection region of the road segment $r_i^a$ and sensing circles of the sensors in $S_i^j$. 

On the path  $PV_i^j$, let $v_i$, and $v_j$ be two consecutive internal vertices and their corresponding sensors are $s_i$ and $s_j$.  Let $s_i$ be a sensor in $S_i^t$, then intersection point between $s_i$'s sensing circle and $l_i^t$ is referred as $p_i^t$. There is another intersection point between sensing circles of $s_i$ and $s_j$, which is inside $r_i^a$. This intersection point is referred as $p_{ij}$. Since $p_i^t$ and $p_{ij}$ both points are inside the convex regions $r_i^a$ and $circle(s_i,\rho)$, therefore the line segment joining $p_i^t$ and  $p_{ij}$ is completely inside $r_i^a$ and $circle(s_i,\rho)$. In this way, it can be shown that for the path $PV_i^j$  in the intersection graph of $r_i$ there exists a piece wise linear path $P_i^j$ between $l_i^t$ to $l_i^b$. The path $P_i^j$ is passing through the intersection points between the sensing circles of the sensors in $S_i^j$ and side boundary of $r_i$. As well as the path $P_i^j$ is passing through the intersection regions of road segment $r_i^a$ and the sensing circles of the sensors in $S_i^j$. Once such path exists then according to lemma \ref{lem:algo2:lem1}, $r_i$ is collaboratively road covered by the sensors in $S_i^j$.
\end{proof}

For example, intersection graph corresponding to the road segment $r_i$ and the sensors deployment in Figure \ref{coll_road_cover} is shown in  Figure \ref{intersection_graph}. For the road segment $r_i$, $S_i^t= \{s_1, s_{10}, s_4, s_7 \} $ and  $S_i^b = \{s_3, s_{11}, s_5, s_6, s_9 \} $. There is a path $PV_i^j= \{v_i^t, v_4, v_5, v_6, v_i^b \}$ in the intersection graph between $v_i^t$ and $v_i^b$ through the internal nodes  $ V_i^j= v_4, v_5$ and $v_6$. Hence,  sensors $S_i^j=  s_4, s_5$ and  $s_6$ collaboratively road cover the road segment $r_i$.

\begin{figure}[h]
\centering
\includegraphics [width=8cm]{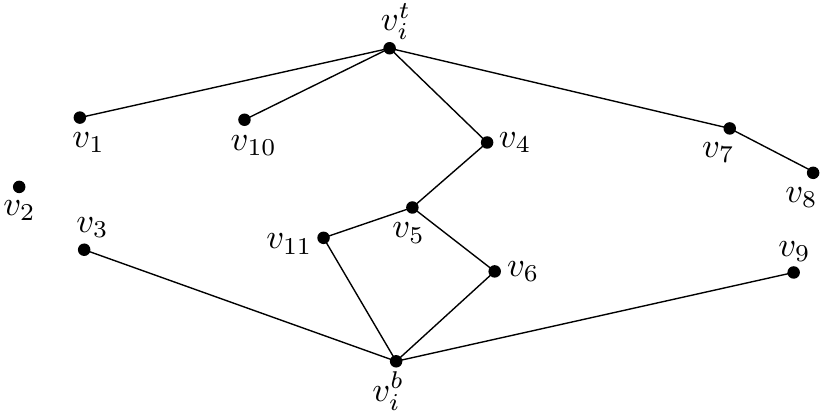}
\caption{Intersection graph/ coverage graph for Figure  \ref{coll_road_cover}}
\label{intersection_graph}
\end{figure}


\begin{theo}
Verification of {\em collaborative road coverage} of a road network with $n$ road segments and $m$ sensors can be done in $O(nm^2)$ time.
\end{theo}

\begin{proof}
Time complexity of collaborative road coverage is measured in two steps: determining an intersection graph and then determining a path between $v_i^t$ and $v_i^b$ for each road segment $r_i :  i \in \{1 \ldots  n \}$. Computation time to find intersection graph corresponding to a road segment is  $O(m^2)$. Once the intersection graph is known, finding a path in the intersection graph for the road segment is linear to the number of edges in the intersection graph which is in worst case $O(m^2)$. Therefore, total time complexity to verify collaborative road coverage for a road network consist of $n$ road segments is $O(nm^2)$.
\end{proof}

\section{Sensors Deployment Schemes to Ensure Road Coverage}
\label{sec:deployment_scheme}

In this section, we discuss sensor deployment schemes to ensure {\em independent road coverage}. We assume that road segments are axis parallel of a given fixed width $w$ and sensors sensing regions are circular disks of a given fixed radius $\rho \ge w$. First, we analyze the complexity of {\bf RCDPL} problem. Thereafter, we discuss two sensor deployment algorithms for two different cases : (i) sensors are allowed to deploy at any arbitrary location, and  (ii) sensors are allowed to deploy only along the side boundaries (top and bottom side boundaries) of the road segments. Before discussing our algorithms in detail, we introduce few terminologies.

\begin{figure}[h]
\centering
\includegraphics [width=8cm]{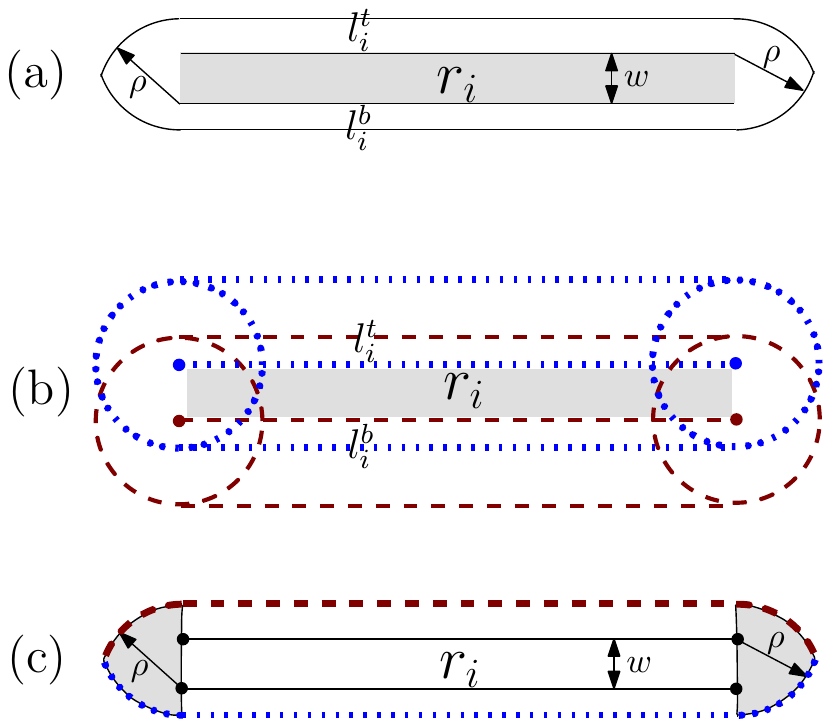}
\caption{(a) Capsule of a road segment $r_i$, (b) Intersection of Minkowski sums, and (c) $Lcap(r_i,\rho)$ and $Rcap(r_i,\rho)$}
\label{road_deploy}
\end{figure}

\begin{definition}{\bf [Capsule :]}
For a given road segment $r_i$ of width $w$ and a positive real number $\rho \ge w$, the capsule $C(r_i,\rho)$ (shown in Figure \ref{road_deploy}(a)) is the intersection of Minkowski sums \cite{compgeom:2000} of a disk of radius $\rho$ on the two segments $l_i^t$ and $l_i^b$ (drawn with black dashed lines and blue dotted lines see Figure \ref{road_deploy}(b)) corresponding to the road segment $r_i$,  which defines the capsule $C(r_i,\rho)$.
\end{definition}

\begin{definition}{\bf [Cap :] }
For a given road segment $r_i$ of width $w$ and a positive real number $\rho \ge w$, cap of the capsule $C(r_i,\rho)$ is left and right part of it as shown by shaded region in Figure  \ref{road_deploy}(c). There are two caps of a capsule $C(r_i,\rho)$ : left-cap $LCap(r_i,\rho)$ and right-cap $RCap(r_i,\rho)$ based on their positions with respect to the capsule $C(r_i,\rho)$.
\end{definition}

\begin{observation} \label{obs:1}
A road segment $r_i$ is independent road covered by a sensor $s_j$ with sensing range $\rho$ if and only if the point sensor $s_j$ is placed inside the capsule $C(r_i,\rho)$ where $\rho \ge w$. 
\end{observation}

\subsection{Complexity results for RCDPL}
\label{ssec:complexity}

Fowler et al. \cite{fowler1981} showed that covering a given set of points in the plane using minimum number of unit disks is NP-hard. Points are a special case of road segments, where lengths and widths of the road segments are zero. So, covering a given set of points in the plane using minimum number of disks is a special case of our problem RCDPL. Hence, RCDPL is NP-Hard. 

Two constant factor approximation algorithms for independent road covering axis parallel road segments are described in the following two subsections.

\subsection{Approximation algorithm for sensor deployment at arbitrary place}
\label{ssec:8-factor}

In this subsection, we present approximation algorithm for sensor deployment to ensure {\em independent road coverage} for axis parallel road segments. We present an 8-factor approximation algorithm for this problem.


First, the road segments in $R$ is divided into two disjoint subsets $R_h$ (denotes set of horizontal road segments) and $R_v$ (denotes set of vertical road segments). The deployment scheme to cover axis parallel road segments is divided into two phases. In the first phase,  sensors are deployed to cover the horizontal road segments in $R_h$, thereafter similar technique is followed to cover the vertical road segments in $R_v$. For the sake of simplicity, we discuss deployment scheme only for horizontal road segments.


\begin{algorithm}
	  
	  $Q_h = \emptyset$,  $Q_v = \emptyset$ \;
	   
	  Classify the road segments in $R$  into two disjoint sets $R_h$ and $R_v$ depending on their orientations.
	  
	  \tcc{Find sensors requirement $Q_h$ for covering horizontal road segments in $R_h$}
	
	  $L_h$ = Sort the road segments in $R_h$ from left to right based on their right end boundary's x-coordinate values\;
	  
	  $I_h = \emptyset$ \;

	  \While { $ L_h \neq \emptyset $}
	  {
	    Select the first road segment $r_i$ from $L_h$ \;
	    $L_h=L_h \setminus r_i$ \;    
	    $I_h = I_h \cup r_i$ \;
	    
	   \While{ $\exists_{r_j} \in L_h \mid (C(r_i,\rho) \cap C(r_j,\rho) \neq \emptyset ) $ }
	    {
	          $L_h=L_h \setminus r_j$ \;
	    }
	    
	   Based on the requirement deploy at most  four sensors $Q_i =\{s_{i1},s_{i2}, s_{i3}, s_{i4}\}$ at the right end boundary of $r_i$, as shown in Figure \ref{fig:hz_road_ind_cover}(c) \;
	    
	    $Q_h= Q_h \cup Q_i$
	  }
	  
	  \tcc{Similarly, find sensors requirement $Q_v$ for covering vertical road segments in $R_v$}
	  
	    $Q= Q_h \cup Q_v$  
	    
	    Return $Q$
\caption{Sensor Deployment Algorithm for Independent Road Coverage of Axis parallel Road Segments $R$}
\label{algo:SensorDeployment}
\end{algorithm}

Sort all the horizontal road segments in increasing order of their right end boundary's x-coordinates and store them in a list $L_h$. Next, select a road segment $r_i$ having left-most right end boundary and add it to another list $I_h$. Put four sensors $Q_i= \{s_{i1},s_{i2}, s_{i3}, s_{i4}\}$ at the right end of $r_i$, as shown in Figure \ref{fig:hz_road_ind_cover}(c). This deployment covers any other road segment $r_j$ which share a position $p$ with $r_i$ such that a sensor positioned at $p$ is able to independent road cover both $r_i$ and $r_j$ together simultaneously. Remove all such road segments $r_j$ from $L_h$, which are covered by these four sensors. Repeat the above process for the remaining road segments in $L_h$ until $L_h$ becomes empty. The detail algorithm for covering axis parallel road segments $R$ is presented in Algorithm \ref{algo:SensorDeployment}.

\begin{lemma}
\label{lem:aa}
Two road segments $r_i$ and $r_j$ are independent road covered by a sensor with sensing range $\rho$  if and only if capsules $C(r_i,\rho)$ and $C(r_j,\rho)$ intersect with each other.
\end{lemma}
\begin{proof}
Road segment $r_i$ can be independent road covered by a sensor $s_k$  iff $s_k$ is placed inside $C(r_i, \rho)$. Similarly, the same sensor $s_k$ covers $r_j$ iff $s_k$ is also inside $C(r_j, \rho)$. Hence there must be a common intersection point between  $C(r_i,\rho)$ and  $C(r_j,\rho)$. 
 \end{proof}

\begin{lemma} 
\label{lem:bb}
For any two road segments $r_i$ and $r_j$ in $L_h$, if $C(r_i,\rho)$ and $C(r_j, \rho)$ intersect with each others and $r_i$ precedes $r_j$ in $L_h$ then some portion of $r_j$'s length but full width is completely inside right-cap $RCap(r_i, 2\rho)$.
\end{lemma}

\begin{proof}
Since $r_i$ precedes $r_j$ in $L_h$, $r_i$ right end boundary's x-coordinate value is less than or  equal to $r_j$'s right end boundary's x-coordinate value. There are two possibilities of the left end boundary of $r_j$ (i) $r_j$'s left end boundary starts before or from the right end boundary of $r_i$ or  (ii) $r_j$'s left end boundary starts after the right end boundary of $r_i$. Therefore, for the case (i) where $r_j$'s left end boundary starts before the right end boundary of $r_i$,  as $C(r_i,\rho)$ and $C(r_j, \rho)$ intersect with each others, distance from both right corner points of $r_i$ to $l_j^t$ and $l_j^b$ are less than or equal to $2\rho$ (shown in Figure \ref{fig:CapsuleCapsule_intersection}(a)). Hence, full height of $r_j$ must be inside the right-cap $RCap(r_i,2\rho)$.  Now consider case (ii), where $r_j$'s left end boundary starts after the right end of $r_i$. Since $C(r_i,\rho)$ and  $C(r_j,\rho)$ intersect with each other, there is a common intersection point between $RCap(r_i, \rho)$ and $LCap(r_j, \rho)$ from where the distance to $r_i$'s right end corner points and distance to left end corner points of $r_j$ are less than or equal to $\rho$, as shown in Figure \ref{fig:CapsuleCapsule_intersection}(b). Hence, according to triangular inequality distance between the farthest pair of corner points of $r_i$'s right end  and $r_j$'s left end  $ \le 2\rho$. Therefore, both corner points of $r_j$'s  left end must be inside the right-cap $RCap(r_i, 2\rho)$ and hence, the left end boundary (full width) of $r_j$ is completely inside $RCap(r_i, 2\rho)$. 
\end{proof}

\begin{figure}[h]
\centering
\includegraphics [width=8cm]{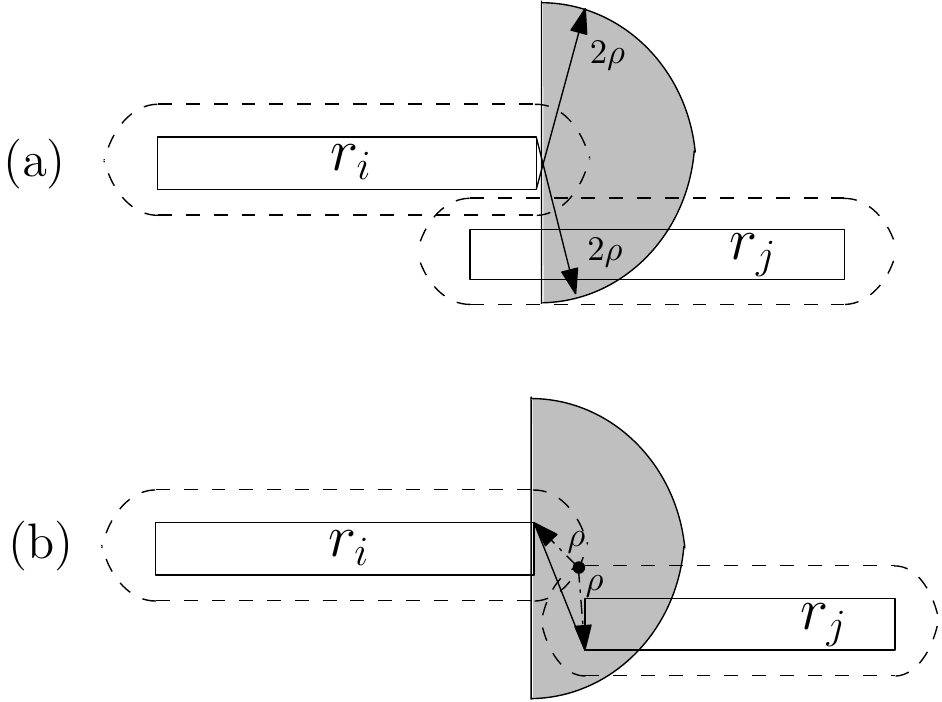}
\caption{Capsule capsule intersection}
\label{fig:CapsuleCapsule_intersection}
\end{figure}


\begin{definition} {\bf [$\delta$-height independent covered:] }
 A region is said to be $\delta$-height independent covered by a set of sensors if within the region at any arbitrary position a vertical line segment of height $\delta$ is always completely inside sensing range of at least one sensor. 
\end{definition}

\begin{lemma}
If $r_i$ is a road segment in $L_h$ of width $w$ and its right end boundary is left most then $w$-height independent covering right-cap $RCap(r_i,2\rho)$ {\em implies independent road covering} $r_i$ together with any road segment $r_j \in L_h$ such that $C(r_i, \rho) \cap C(r_j, \rho) \ne \emptyset$.
\end{lemma}


\begin{proof}
According to Lemma \ref{lem:aa}, road segments $r_i$ and any other road segment $r_j$ can be  independently road covered by a sensor if and only if $C(r_i,\rho)$ and  $C(r_j,\rho)$ intersect with each other. Again $C(r_i,\rho)$ and  $C(r_j,\rho)$ intersect with each other and $r_i$'s right end boundary is leftmost, therefore according to Lemma \ref{lem:bb}, some part of $r_j$'s length but full width of $r_j$ must be completely inside right-cap $RCap(r_i,2\rho)$ (shaded region in Figure \ref{fig:hz_road_ind_cover}(b) for $r_i =r_1$). Therefore, if right-cap $RCap(r_i,2\rho)$ is $w$-height independent covered then road segment $r_i$ together with all other road segments whose capsules intersect $C(r_i,\rho)$ is automatically {\em independent road covered}. 

\end{proof}

\begin{lemma}
\label{lem:cc}
Four sensors with sensing range $\rho \ge w$ are sufficient to $w$-height independent covering right-cap $RCap(r_i,2\rho)$ of road segment $r_i$ of width $w$.
\end{lemma}

\begin{proof}
If four sensors $\{s_{i1}, s_{i2},s_{i3}, s_{i4} \}$ are placed as in Figure \ref{fig:hz_road_ind_cover}(c), then right-cap $RCap(r_i,2\rho)$ is $w$-height independent covered. The detail proof is discussed in Appendix section.
\end{proof}

\begin{figure}[h]
\centering
\includegraphics [width=8cm]{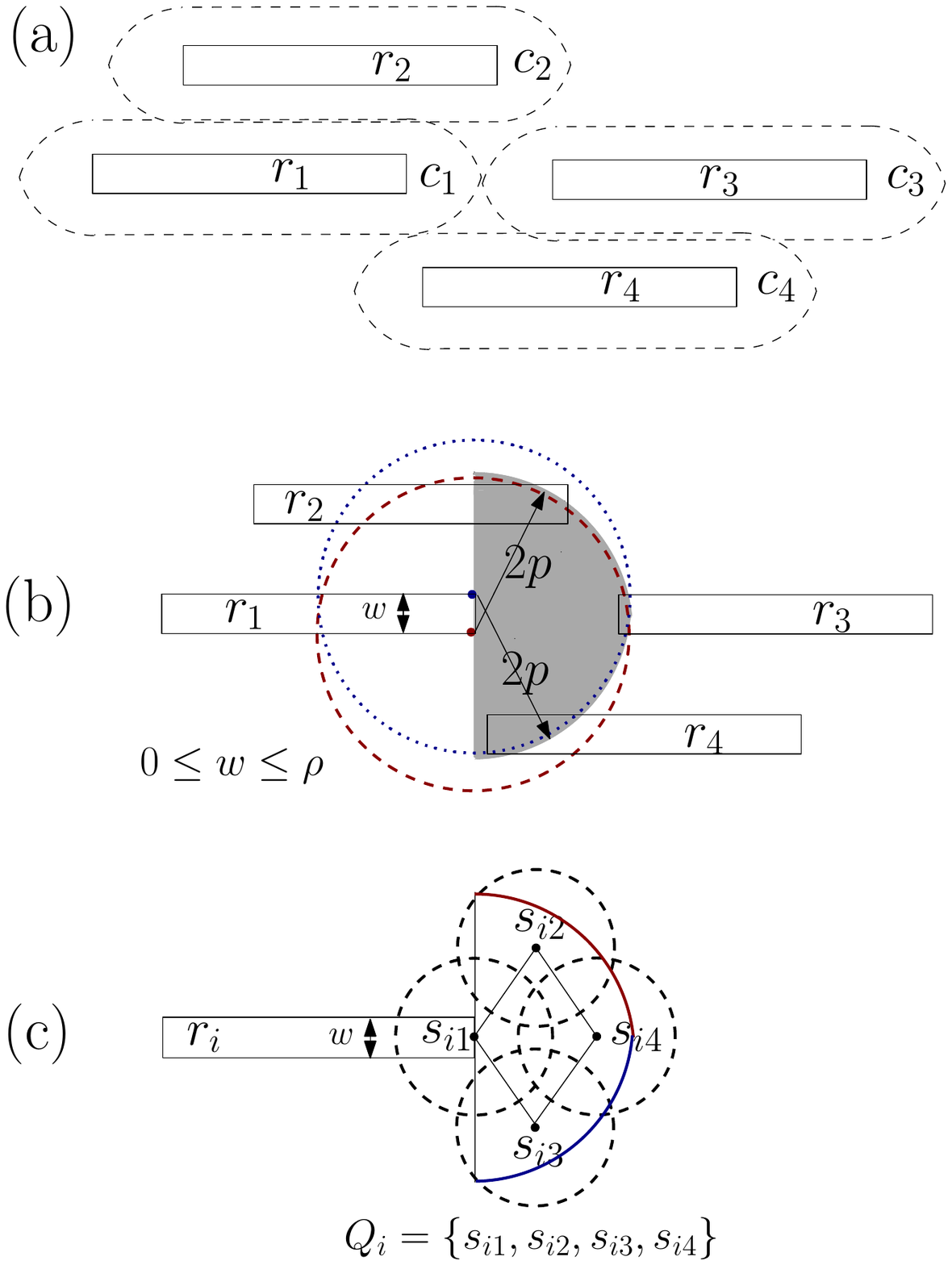}
\caption{Independent road covering horizontal road segments}
\label{fig:hz_road_ind_cover}
\end{figure}

\begin{theo}
\label{th:gg}
All axis parallel road segments can be independent road covered in $O(n^2)$ time and the number of sensors used is $\le 8 OPT$
\end{theo}

\begin{proof}
According to Algorithm \ref{algo:SensorDeployment}, at least one sensor is required to independent road cover a road segment $r_i \in I_h$. The capsules corresponding road segments in $I_h$ are disjoint. Therefore, to independent road cover all the road segments in $R_h$  at least $|I_h|$ sensors are required,  because $I_h \subseteq R_h$. But, our algorithm uses at most $|Q_h|=4|I_h|$ sensors to cover all the road segment in $R_h$. Hence, our algorithm uses at most four times more than the optimal number of sensors require to independent road cover all horizontal road segments in $R_h$. Same process is repeated for covering vertical road segments in $R_v$. Therefore, $|Q_v|=4|I_v|$ sensors are used to cover vertical road segments in $R_v$. Let $OPT$ denote optimal number of sensors require to cover all the road segments in $R$. Therefore, $OPT \ge Max(|I_h|, |I_v|)$.  Our algorithm uses in total $|Q| = |Q_h| + |Q_v|= 4(|I_h|+|I_v|) \le 8 Max(|I_h|, |I_v|) \le  8 OPT$ sensors. Hence, in total our algorithm uses at most eight times more than the minimum number of sensors required to cover all axis parallel road segments in $R$. 

There is a nested while loop in Algorithm \ref{algo:SensorDeployment}, which runs at most $n^2$ times, where $n$ denotes the number of road segments. Therefore, the time complexity of the algorithm is $O(n^2)$.
\end{proof}

Although a road segment is get covered if a sensor is placed anywhere inside the bounding capsule. But, in practice the ends of road segments are junctions (start or end of the road). Therefore, in general sensors are not deployed on the roads or end of the roads. In the next section, we discuss sensor deployment algorithm to cover road segments where sensors are deployed only along the side boundaries of the road segments ($l_i^t$ or $l_i^b$ for  road segment $r_i$).

\remove{
 
%
%
%
}


\subsection{Approximation algorithm for sensor deployment along the side boundary}
\label{ssec:4-factor}

In this subsection, we present {\em independent road covering} algorithm for axis parallel road segments, where sensors are placed only along the side boundaries of the road segments. We assume sensors sensing regions are disks of equal radius $\rho \ge w$. First, we describe a 2-factor approximation algorithm for horizontal road segments. The same technique is used to cover vertical road segments as in the previous algorithm. These two solutions are combined to {\em independent  road cover} all axis parallel road segments and together produces a 4-factor approximation algorithm. The algorithm for road covering horizontal road segments works as follows:

Given a set of  horizontal road segments $R_h$. First sort the road segments in $R_h$ according to their right end boundary's $x$-coordinates value and store them in a list named $L_h$. Next, select a road segment $r_i$ having left-most right end boundary. Put $r_i$ in another list $I_h$, and then remove $r_i$ from $L_h$ and any other road segment $r_j$ from $L_h$, which intersects $RCap(r_i, \rho)$ or $RCap(r_i, \rho) \cap r_i^a \ne \emptyset$ . In others words, remove all the road segments from $L_h$ which are independent {\em road covered} by any one of the two sensors  deployed at top and bottom right corners of $r_i$. Place two sensors at the two right corners of the road segment $r_i$ as in Figure \ref{fig:coveringRoads}(b) and call them as  $Q_i= \{s_{i1}, s_{i2} \}$. Sensor used by the algorithm are stored cumulatively in a list called $Q_h$,  which is updated in each iteration by $Q_h= Q_h \cup Q_i$. Repeat the above process for the remaining road segments in $L_h$ until $L_h$ becomes empty.

\begin{lemma} 
\label{lem:ee}
Two horizontal road segments $r_i$ and $r_j$ can be independent road covered by a single sensor $s_k$ with sensing range $\rho$ if and only if $r_j$'s top side boundary $l_j^t$  or bottom side boundary $l_j^b$ intersects the capsule $C(r_i,\rho)$ or vice versa. Assuming sensors can be placed only on the top or bottom side boundary of the horizontal road segments.
\end{lemma}

\begin{proof}
According to observation \ref{obs:1}, road segment $r_i$ can be road covered by sensor $s_k$ when  sensor $s_k$ is placed inside the capsule $C(r_i,\rho)$. Sensors are restricted  to place only on the side boundary of the road segments. Therefore, to cover both $r_i$ and $r_j$ by a single sensor $s_k$, either $s_k$ is placed on the $l_i^t$ or $l_i^b$ and $s_k$ must be inside $C(r_j,\rho)$ or vice versa. Hence, $C(r_j,\rho)$ must intersects one of the two side boundaries of $r_i$. Similarly, when $s_k$ is placed on one of the two side boundaries of $r_j$ then $C(r_i, \rho)$ must intersects one of the two side boundaries of $r_j$. It proves the lemma.
\end{proof}

\begin{lemma}
\label{lem:dd}
 If $l_j^t$ or $l_j^b$ of a road segment $r_j$ intersects right-cap $RCap(r_i,\rho)$ then both $l_j^t$ and $l_j^b$  of $r_j$ intersect one of the circles centered at two right corners of $r_i$ with radius $\rho$. 
\end{lemma}

\begin{proof}
Assume $circle(\alpha,\rho)$ denotes a circle centered at point $\alpha$ with radius $\rho$. For the sake of contradiction assume that the bottom boundary of $r_j$ intersects  $RCap(r_i,\rho)$ at point $e$, as shown in Figure \ref{fig:coveringRoads}(b), but its top boundary does not intersect the $circle (b,\rho)$. Note arc $\arc{ob'}$ is part of $RCap(r_i,\rho)$  and $circle(c,\rho)$ and point $e$ is inside the $circle (b,\rho)$. Let $f$ be a point on the top boundary of the road segment $r_j$ which is on the perpendicular direction of $e$. Therefore, line segments $bc$ is parallel to $ef$ and is of equal length. Hence $ce=bf=\rho$. Therefore, $f$ must be on the circumference of $circle(b,\rho)$ and $e$ is already inside $circle(b,\rho)$. Therefore, both boundaries of $r_j$ intersect the circle centered at $b$, which contradicts our assumption. Hence the lemma is true.
\end{proof}
 
\begin{lemma} 
\label{lem:ff}

Let $r_i$ be a road segment in $L_h$, whose right end boundary is left most. Set of all horizontal road segments, whose side boundaries intersect the capsule $C(r_i,\rho)$, together with $r_i$ can be {\em independent road covered} by using only two sensors placed on the two right corner points of $r_i$. Assuming sensors can be placed only on side boundary of road segments.

\end{lemma}

\begin{proof}
Two sensors $s_{i1}$ and $s_{i2}$ are deployed at the two right corners of $r_i$ ($b$ and $c$), as shown in Figure \ref{fig:coveringRoads}(b). This is easy to follow that they are able to {\em independent road cover} the road segment $r_i$, since $w \le \rho$. On the other side, since $r_i$ has left most right end boundary, therefore all other road segments whose side boundaries intersect the capsule $C(r_i,\rho)$ must intersect right-cap $RCap(r_i,\rho)$ as shown in Figure \ref{fig:coveringRoads}(b). $RCap(r_i,\rho)$ consists of two circular arcs : upper arc $\arc{ob'}$ centered at $c$ and lower arc $\arc{oc'}$ centered at $b$. According to Lemma \ref{lem:dd}, if $l_j^b$ of  road segment $r_j$ intersects the  $RCap(r_i,\rho)$  then both $l_j^t$  and $l_j^b$ of $r_j$ also intersect one of the circles centered at $b$, $c$. Since upper arc $\arc{ob'}$ is part of $RCap(r_i,\rho)$, therefore both $l_j^t$  and $l_j^b$  intersect one of $circle(b,\rho)$, $circle(c,\rho)$.  This is trivial to show that if $l_j^b$ intersects arc $\arc{ob'}$  then both $l_j^t$  and $l_j^b$ of $r_j$ also intersect $circle(b,\rho)$  centered at $b$. Similar argument is applicable for the intersection of $l_j^t$  with the lower arc $\arc{oc'}$. Therefore, if two sensors are placed on the two right corners of $r_i$ then all the road segments, whose side boundaries intersect the capsule $C(r_i,\rho)$, can be independently road covered by one of the two sensors.
\end{proof}

\begin{figure}[t]
\centering
\includegraphics [width=8cm]{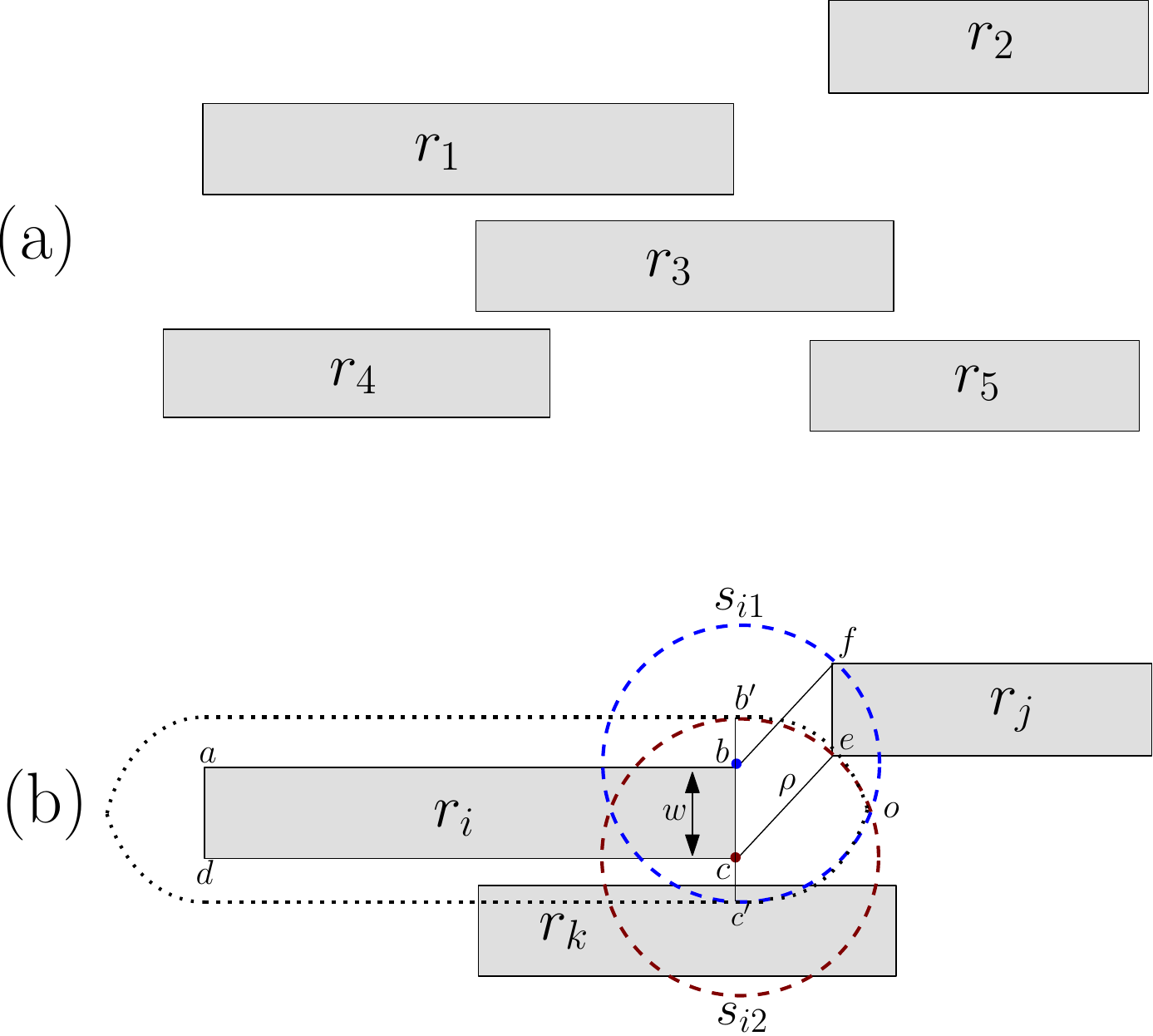}
\caption{(a) Horizontal road segments, (b) Processing of the road segment $r_i=r_1$, $r_j=r_2$, $r_k=r_3$ and $Q_i= \{ s_{i1}, s_{i2} \}$}
\label{fig:coveringRoads}
\end{figure}

\begin{lemma}
All horizontal road segments can be {\em independent road covered}  in $O(n^2)$ time and the number of sensors used is $\le 2 OPT$, for the case where sensors are allowed to place only on the bounding sides of the road segments. 
\end{lemma}

\begin{proof}
According to our algorithm and Lemma \ref{lem:ff} no two road segments in $I_h$ can be {\em independently road covered} by a  sensor. Therefore, at least one sensor is required to {\em independent road cover} a road segment in $I_h$. Hence, in total at least $|I_h|$ sensors are necessary to cover all the road segments in $R_h$. Let $OPT$ denote the optimum number of sensors required to cover all the road segments in $R_h$. Since $I_h \subseteq R_h$ then $OPT \ge |I_h|$. According to our algorithm, two sensors are deployed at top and bottom right corners for each road segments in $I_h$. Therefore, in total $2|I_h|$ sensors are deployed. Hence, number of sensors used in our algorithm is $ |Q_h|= 2|I_h| \le 2OPT$. Sorting all the road segments in $R_h$ takes $O(nlogn)$ time. Thereafter, our algorithm finds for each road $r_i \in I_h$ set of road segments in $R_h$ which are intersecting the right-cap $RCap(r_i,\rho)$. This step requires $|I_h| \times |R_h| = O(n^2)$ time. Therefore, total time complexity of our algorithm is $O(n^2)$.
\end{proof}

\begin{theo}
All axis parallel road segments in $R$ can be independent road covered in $O(n^2)$ time and the number of sensors used is $\le 4 OPT$
\end{theo}

\begin{proof}
Similar algorithm for covering horizontal road segment is applicable for covering vertical road segments in $R_v$ . The sensor used for covering vertical road segments are stored in $Q_v$, where $|Q_v|= 2 |I_v|$.  Therefore, total number of sensors used for covering all axis parallel road segment in $R$ is $Q= Q_h+Q_v$. The rest of the proof is almost similar to the proof of Theorem \ref{th:gg}.
\end{proof}


\section{Simulation Results}
\label{sec:simulation}
In this section, we study the performance of our proposed algorithms for independent road coverage problem through simulation. We have designed simulator in MATLAB to implement our proposed algorithms. We compare solutions returned by the two approximation algorithms. For simplicity, we have considered only horizontal road segments throughout the simulation. 

Sample output of our simulator for independent road coverage are shown in Figure \ref{fig:result1}, \ref{fig:result2}. Figure \ref{fig:result1} shows the deployment of sensors of our simulator, where sensors are allowed to place at arbitrary positions. While Figure \ref{fig:result2} shows deployment of sensors on the side boundary of the road segments. In both figures, an instance of $20$ horizontal road segments is considered and sensing radius of the sensors are set to $75$.

\begin{figure}[t]
\centering
\includegraphics [width=8cm]{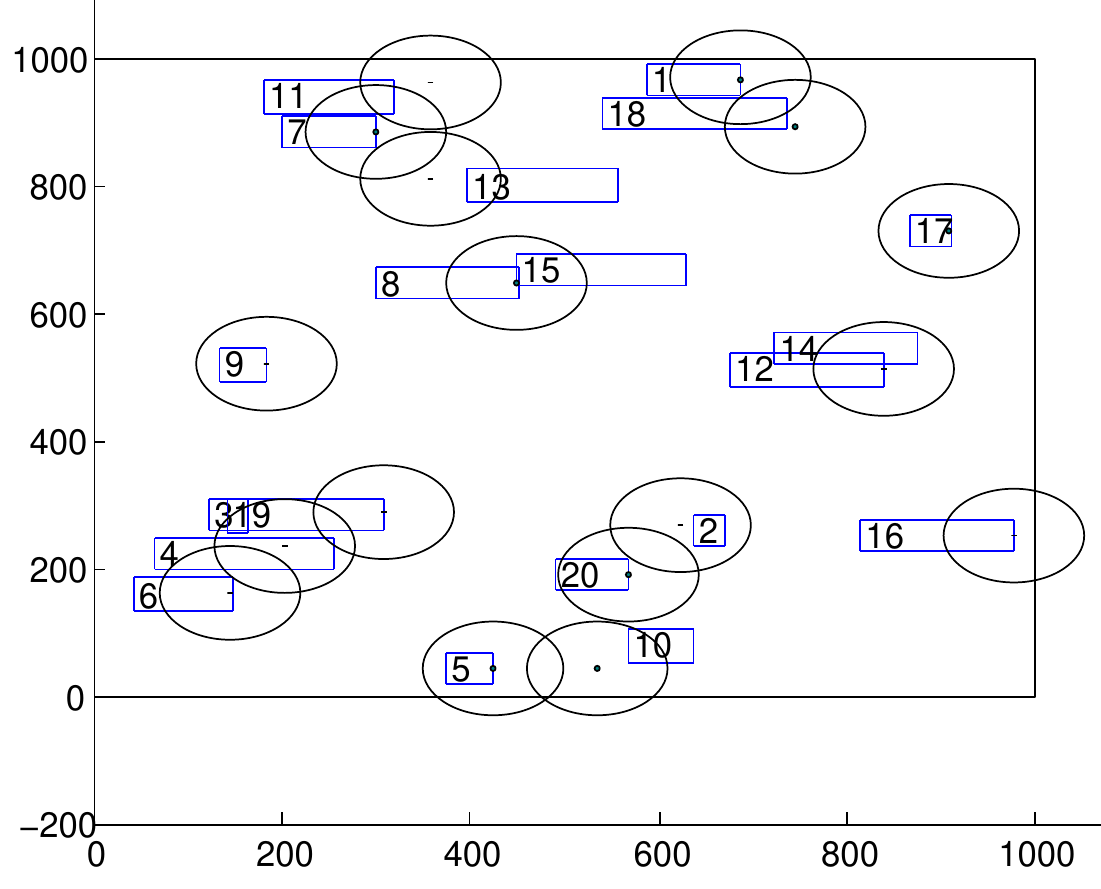}
\caption{An output of our simulator for road coverage where sensors can be deployed at any arbitrary position }
\label{fig:result1}
\end{figure}

\begin{figure}[t]
\centering
\includegraphics [width=8cm]{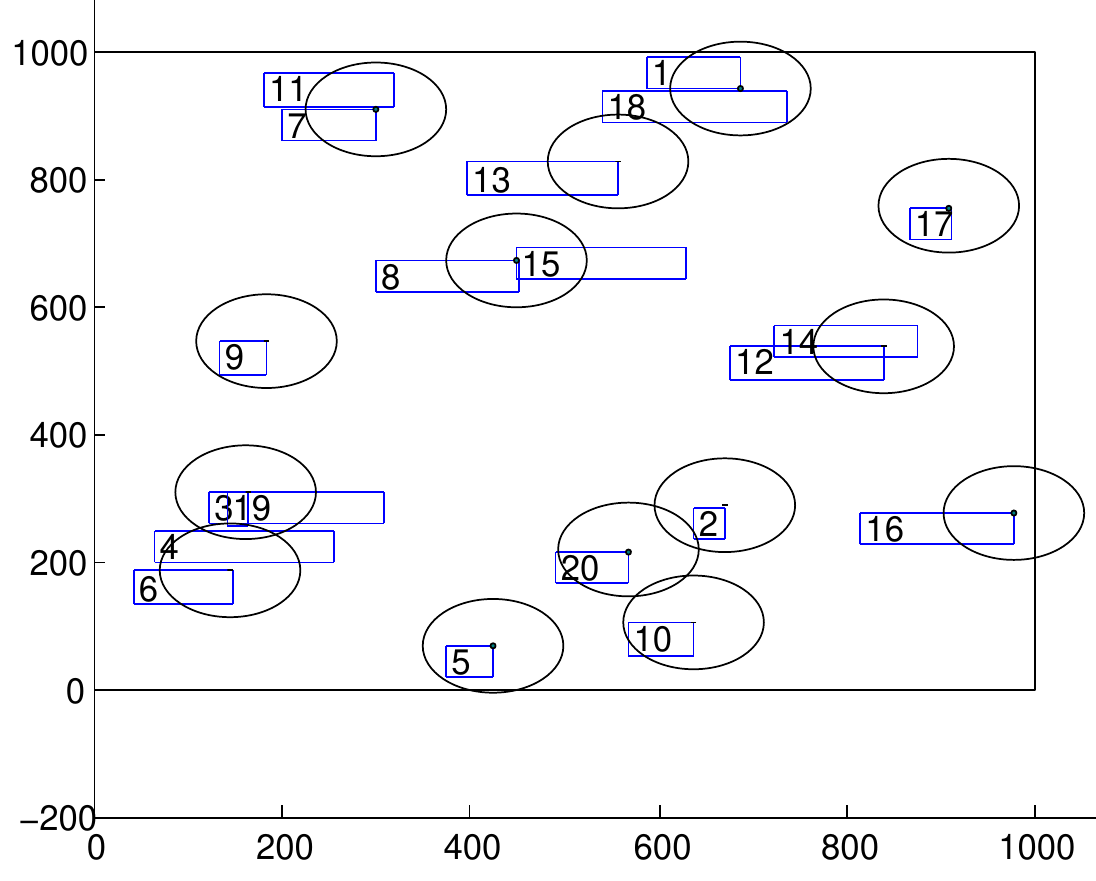}
\caption{An output of our simulator for road coverage where sensors can be deployed only on the side boundary }
\label{fig:result2}
\end{figure}

We have simulated both algorithms discussed in section \ref{sec:deployment_scheme} for finding suitable sensor deployment strategy for independent road coverage. We are referring the algorithm where there is no restriction on sensors positions as {\bf arbitrary}. Similarly in {\bf side boundary} sensor deployment scheme,  sensors are deployed only on the side boundary of the road segments. We evaluate how the number of sensors used in the algorithms vary with the number of road segments and sensing range of the sensors. Moreover, lower bounds of the algorithms are also reported beside the actual sensor used. In both algorithms lower bounds are computed by measuring the size of the list $I_h$ for horizontal road segments. During simulations, road segments  are placed randomly within a rectangular region of size $1000 \times 1000$. Width of the road segments are set to $50$.  Length of the road segments are randomly picked and varied within $0$ to $200$. Outputs of the algorithms for different number of road segments with sensing radius  $\rho = 75$ and $\rho = 100$ are shown in Table \ref{tab:tab1} and Table \ref{tab:tab2} respectively. The results reported in the tables are average of $50$ independent runs. Although the approximation ratio of the arbitrary sensor deployment algorithm is higher than side boundary sensor deployment algorithm, but simulations results shows that in practice the difference between them is very less. From the two tables it is observed that as the sensing radius of the sensors are increased, number of sensors requirement decreases proportionally. We evaluate the sensors requirements by varying the number of road segments for $20$, $30$ and $40$. We evaluate the lower bound and actual number of sensors used by our proposed algorithms to independent road cover all the road segments for each instance. From the results it is found that actual sensors used by the algorithms are much lesser than their theoretical estimations (upper bounds).

\begin{table}[!ht]
\centering

    \caption{Sensing Radius $\rho= 75$}     
    \label{tab:table}

    \begin{tabular}{|l|l|l|l|l| }
    \hline
    {\bf Road } & \multicolumn{2} {c|} {\bf  Side Boundary } & \multicolumn{2} {c|} {\bf Arbitrary} \\
     
    {\bf Segments} &  Lower  &   Sensors  &  Lower  &  Sensors \\
                              &  Bound  &  Deployed  &   Bound & Deployed  \\
    
    \hline 
    20        &  14.12 & 14.58  &  10.18 & 16.32   \\
    \hline
   30         & 18.36 &  19.30  & 12.46 &  22.58     \\
    \hline
   40         &  21.66 &  23.94  &  14.06 & 28.48     \\
    \hline
    \end{tabular}
   
    \label{tab:tab1}
\end{table}


\begin{table} [h!]
\centering
\caption{Sensing Radius $\rho= 100$}

    \begin{tabular}{|l|l|l|l|l| }
    \hline
    {\bf Road } & \multicolumn{2} {c|} {\bf  Side Boundary } & \multicolumn{2} {c|} {\bf Arbitrary} \\
     
    {\bf Segments} &  Lower   &  Sensors     &  Lower  &  Sensors \\
                              &  Bound  &  Deployed  &   Bound &  Deployed \\
   \hline

20 & 12.48 & 13.08  & 7.72  & 14.92   \\ 
30 & 15.80 & 17.16  & 9.18  & 19.78   \\ 
40 & 18.08 & 20.94  & 9.88  & 23.42   \\  
\hline
\end{tabular}
\label{tab:tab2}
\end{table}

\section{Conclusion}
\label{sec:conclude}
In this paper, a new coverage measure called {\em road coverage} along with its two variations {\em independent road coverage} and {\em collaborative road coverage} are introduced. Algorithms are proposed to verify  {\em road coverage}. Although we have assumed that the road  segments are rectangular shaped and sensing regions of the sensors are circles but our coverage measuring algorithms work for any quadrilateral shaped road segments as well as any convex shaped sensing regions. Only one requirement is to have an algorithm for finding intersection between the two shapes. The problem of covering a set of road segments with minimum number of sensors is addressed in this paper. The problem is shown NP-hard. An $O(n^2)$ time $8$-factor approximation algorithm is proposed, where the road segments are axis-parallel. We also present an $4$-factor approximation algorithms for a special case, where sensors are allowed deploy only along the side boundaries of the road segments. Experimental results are also presented. Although the approximation factors are high but in practice solutions return by the algorithms are close to their lower bounds.  Developing efficient algorithms for sensor deployment with good approximation factor for collaborative road coverage and  for the general problem, where the road segments are of arbitrary orientation remains an open challenge. A mild variation is to cover at least a constant fraction of the length of each road segment.


\section{Acknowledgements}
This work is supported by the Science \& Engineering  Research Board, DST, Govt. of India [Grant numbers:  ECR/2016/001035 ];
 
\bibliographystyle{abbrv}
\bibliography{RoadCoverage}

\appendix {\bf Proof of Lemma \ref{lem:cc} :}

\begin{proof} 

\begin{figure}[h]
\centering
\includegraphics[width=7cm]{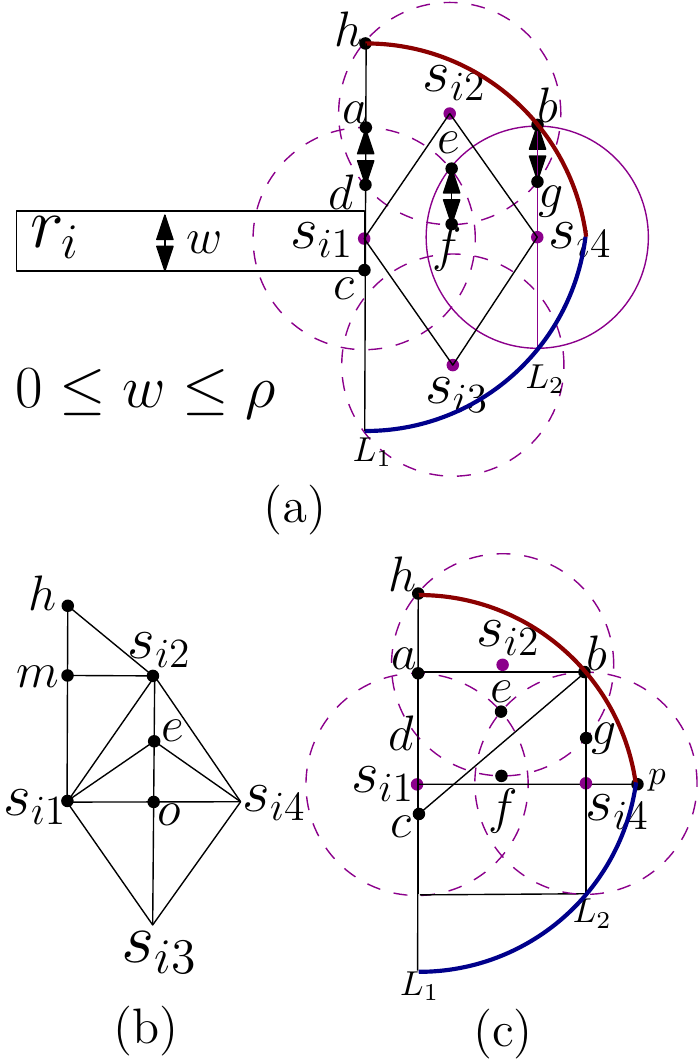}
\caption{ $w$-height independent covering $RCap(r_i,2\rho)$ }
\label{independent_road_cov}
\end{figure}

Four sensors $\{s_{i1}, s_{i2},s_{i3}, s_{i4} \}$  are placed as shown in Figure \ref{independent_road_cov}(a) to $w$-height independent cover right-cap $RCap(r_i,2\rho)$. Sensor $s_{i1}$ is placed at the middle of right end boundary of $r_i$.  And $s_{i4}$ is placed horizontally $\sqrt{(2p)^2 - (p+\frac{w}{2})^2}$ apart from $s_{i1}$. And $s_{i2}$, $s_{i3}$ are placed on the perpendicular bisector of $s_{i1}$ and $s_{i4}$ symmetrically in equal distance from $s_{i1},s_{i4}$.  Let $e$ be an intersecting point between $circle(s_{i1},\rho)$ and $circle(s_{i4},\rho)$, and $f$ denote an intersecting point of $circle(s_{i2},\rho)$ with the perpendicular bisector of $s_{i1}s_{i4}$. The placement of $s_{i2}$ is such that the length of $ef=w$.

 Let $L_1$ denote a vertical line passing through right end boundary of $r_i$,  which intersects $circle(s_{i1}, \rho)$ at $a$. And $h$ and $d$ are  intersection points of $circle(s_{i2}, \rho)$ with $L_1$.  Similarly, let a vertical line $L_2$ passing through $s_{i4}$ intersect $circle(s_{i4},\rho)$ at $b$.  Let $L_2$ intersect $circle(s_{i2},\rho)$ at point $g$ . From Figure \ref{independent_road_cov}(a) and (c), it is obvious that  $ac=\rho+\frac{w}{2}$ and $ab=s_{i1}s_{i4}= \sqrt{(2p)^2 - (p+\frac{w}{2})^2}$.

In Figure  \ref{independent_road_cov}(c), $ bc =\sqrt{ ab^2 + ac^2}= \sqrt{(2p)^2 - (p+\frac{w}{2})^2  +  (p+\frac{w}{2})^2 } = 2\rho$ and hence $b$ is on the perimeter of $RCap(r_i,2\rho)$. 

Therefore, the vertical line $L_2$ passing through $s_{i4}$ intersects $circle(s_{i4},\rho)$ and $RCap(r_i,2\rho)$ at point $b$ and hence $ad=bg$.

 The covering pattern of $RCap(r_i,2\rho)$ by the circles corresponding to the sensors is symmetric with respect to a horizontal line passing through $s_{i1}s_{i4}$. We show that the upper half of right-cap $RCap(r_i,2\rho)$ (region defined by the points $s_{i1},p,b,h$ in Figure \ref{independent_road_cov}(c) ) is $w$-height independent covered by the three circles corresponding to sensors $s_{i1}, s_{i2}$ and $s_{i4}$. Consider intersection region between right-cap of $RCap(r_i,2\rho)$ and sensing circle of $s_{i2}$ :  ($X=RCap(r_i,2\rho) \cap circle(s_{i2}, \rho) )$, if a vertical segment of height $w$ is inside $X$  then it is always under the sensing range of $s_{i2}$. Since $s_{i2}$ is placed on the perpendicular bisector of $s_{i1}s_{i4}$ such that the length of the vertical segment $ef$ is $w$. 

  $\therefore es_{i2} = (fs_{i2} - ef) = (\rho - w ) $ 

 From Figure \ref{independent_road_cov}(c), it is easy to follow that  if $ad, ef, bg \ge w$ and $hc \ge 2\rho$ then the upper half of right-cap $RCap(r_i,2\rho)$  is $w$-height independent covered by the three sensors $s_{i1}$, $s_{i2}$ and $s_{i4}$.  According to the placement of $s_{i2}$, the length of $ef$ is $w$. Therefore, we have to show the length of the remaining three  segments $ad, bg$ and $hc$. In Figure \ref{independent_road_cov}(b), sensors are represented as vertices. A horizontal and a vertical line segments are drawn from $s_{i2}$ which intersect the segment $s_{i1}h$ at $m$, and $s_{i1}s_{i4}$ at $o$ respectively.

$hm = \sqrt{s_{i2}h^2-s_{i2}m^2} = \sqrt{ s_{i2}h^2 - (\frac{s_{i1}s_{i4}}{2})^2 }$ 

$\therefore hm=\sqrt{ \rho^2- \frac{(2p)^2-(p+\frac{w}{2})^2}{4} }= \frac{2\rho + w}{4}$,

The length of the chord $hd$ in Figure \ref{independent_road_cov}(a) is -

$hd= 2hm  = \frac{2\rho + w}{2}$

$oe= \sqrt{s_{i1}e^2 - s_{i1}o^2} = \sqrt{\rho^2 -(\frac{s_{i1}s_{i4}}{2})^2}= \frac{2\rho + w}{4}$

$s_{i1}h = s_{i1}m + hm = os_{i2} + hm = oe + es_{i2} + hm $

$\implies s_{i1}h= oe + hm + es_{i2} = \frac{2\rho + w}{2} + (\rho-w) = 2\rho - \frac{w}{2}$,

$\therefore ad=s_{i1}a+ hd - s_{i1}h = \rho + \frac{2\rho + w}{2} - (2\rho - \frac{w}{2}) =w$

$\therefore hc= s_{i1}h + s_{i1}c= 2\rho -\frac{w}{2}+\frac{w}{2} = 2\rho$ 

Hence, the upper half of right-cap $RCap(r_i,2\rho)$ is $w$-height independent covered by the three sensors $s_{i1}$, $s_{i2}$ and $s_{i4}$. Similarly, it can be shown that lower half of $RCap(r_i, 2\rho)$ is also $w$-height independent covered by $s_{i1}$, $s_{i3}$ and $s_{i4}$. Therefore, $RCap(r_i, 2\rho)$ is $w$-height independent covered by $s_{i1}, s_{i2}, s_{i3}$ and $s_{i4}$.

\end{proof}
\begin{IEEEbiography}[{\includegraphics[width=1in, height=1.25in, clip, keepaspectratio]{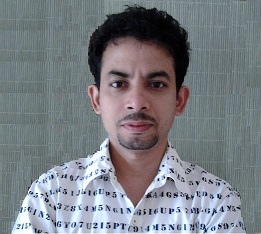}}]{Dr. Dinesh Dash} received Master of Technology in Computer Science and Engineering from University of Calcutta, India in 2004. From 2004 to 2007 he worked as a Lecturer at Asansol Engineering College under West Bengal University of Technology, India. From 2008 to 2012 he worked as a research fellow in the Department of CSE, Indian Institute of Technology Kharagpur, India. His PhD research topics was on coverage problem in Wireless Sensor Network. He was awarded Ph.D. in 2013  from Indian Institute of Technology Kharagpur. He worked as senior research associate from 2013 to 2014 at Infosys Limited, India. From 2013 to 2014 he worked as Assistant Professor at Tezpur University, Assam, India. Since Dec 2014, he is working as an Assistant Professor in the Dept of CSE, NIT Patna. His current work focuses on sensor network coverage problem, data gathering problem, design of fault tolerant system and mobile AdHoc Network.
\end{IEEEbiography}

\end{document}